\pdfoutput=1
\documentclass[showpacs,superscriptaddress,aps,pra, twocolumn]{revtex4-1}
\usepackage{caption}
\usepackage{subcaption}
\usepackage{cancel}
\usepackage{amsfonts}
\usepackage{amsmath,amsthm}
\usepackage{bm}
\usepackage{amssymb}
\usepackage{mathrsfs}
\usepackage{amsmath, amsthm, amssymb,amsfonts,mathbbol,amstext}
\usepackage{mathtools}
\usepackage[colorlinks=true,citecolor=blue,urlcolor=blue]{hyperref}
\usepackage[normalem]{ulem}
\usepackage{pgf,tikz,pgfplots}
\pgfplotsset{compat=1.15}
\usepackage{mathrsfs}
\usetikzlibrary{arrows}
\usetikzlibrary{decorations.pathreplacing}
\usepackage{dsfont}
\usepackage{graphicx}

\newtheorem{thm}{Theorem}
\newtheorem{cor}[thm]{Corollary}

\newtheorem{prop}[thm]{Proposition}

\newcommand{\ket}[1]{\left| #1 \right\rangle}
\newcommand{\bra}[1]{\left\langle #1 \right|}
\newcommand{\ketbra}[2]{\left|#1\middle\rangle\middle\langle#2\right|}
\newcommand{\de}[1]{\left( #1 \right)}

\newcommand{\DE}[1]{\left\{#1\right\}}

\newcommand{\ie}{\emph{i.e.}: }

\newcommand{\mc}[1]{\mathcal{#1}}

\def\be{\begin{equation}}
\def\ee{\end{equation}}
\usepackage{color}
\definecolor{violeta}{cmyk}{0.07,0.90,0,0.34}

\newcommand{\CG}{\Lambda} 
\newcommand{\Dyn}{\mathbf{U}} 
\newcommand{\tr}{\mbox{tr}} 

\usepackage{cancel}
\usepackage{soul}
\definecolor{cgreen}{RGB}{26, 199, 76}

\begin{document}


\title{Investigating Coarse-Grainings and Emergent Quantum Dynamics with Four Mathematical Perspectives}

\author{Cristhiano Duarte} 
\affiliation{Schmid College of Science and Technology, Chapman University, One 
University Drive, Orange, CA, 92866, USA}
\affiliation{Institute for Quantum Studies, Chapman University, One 
University Drive, Orange, CA, 92866, USA}
\affiliation{Wigner Research Centre for Physics, H-1121, Budapest Hungary}

\author{Barbara Amaral}
\affiliation{Departamento de F\'isica e Matem\'atica, CAP - Universidade Federal de S\~ao Jo\~ao del-Rei, 36.420-000, Ouro Branco, MG, Brazil} 
\affiliation{Department of Mathematical Physics, Institute of Physics, University of S\~ao Paulo, R. do Mat\~ao 1371, S\~ao Paulo 05508-090, SP, Brazil}

\author{Marcelo {Terra Cunha}} 
\affiliation{Universidade Estadual de Campinas, Cidade Universit\'aria Zeferino Vaz, 651 S\'ergio Buarque de Holanda, Campinas, SP,13083059, Brazil}

\author{Matthew Leifer}
\affiliation{Schmid College of Science and Technology, Chapman University, One 
University Drive, Orange, CA, 92866, USA}
\affiliation{Institute for Quantum Studies, Chapman University, One 
University Drive, Orange, CA, 92866, USA}

\begin{abstract}
With the birth of  quantum information  science,  many  tools  have been  developed  to  deal with  many-body  quantum  systems. Although a complete description of such systems is desirable, it will not always be possible to achieve this goal, as the complexity of such description tends to increase with the number of particles. 
It is thus crucial to build effective quantum theories aiming to understand how the description in one scale emerges from the description of a deeper scale. This contribution explores different mathematical tools to the study of emergent effective dynamics in scenarios where a system is subject to a unitary evolution and the coarse-grained description of it is given by a CPTP map taking the original system into an \emph{effective} Hilbert space of smaller dimension. We see that a well-defined effective dynamics can only be defined when  some  sort  of matching between the underlying unitary and the coarse-graining map is satisfied. Our main goal is to use these different tools to derive necessary and sufficient conditions for this matching in the general case.
\end{abstract}

\maketitle


\section{Introduction} \label{Sec.Intro}

With the development of technology for control of many-body quantum systems, we are able to build and describe quantum devices with an increasing number of individual systems~\cite{QLYN07,CGP15, GRE14,AruteEtAl19}. Although a complete description of such systems is desirable, it will not always be possible to achieve this goal, since the complexity of such description increases exponentially with the number of particles. Let alone the inherent laboratory errors every experimentalist has to deal with~\cite{Kuhr16}. 
It is crucial, then, to the development of new quantum technologies to build \emph{effective theories}, a machinery to understand how  the description in one scale emerges from the description in a deeper scale, an idea that is central not only in (quantum) physics~\cite{Kabernik18, Wolfram83, Castiglione08} but predominant across many different scientific fields~\cite{ZTZ13,PSP13,Baeurle09}.

Suppose one wants to describe the closed, unitary evolution~\cite{NC00} a quantum system goes through. It might be the case that -- either due to imperfections or lack of full knowledge or the high-complexity of the underlying system -- one does not have full access to that system and, instead, the observer only has access to or can interact with  a less informative part of the original system. Suppose that even in this extremely adverse situation, one still wants to describe the evolution of the system they have access to. This rather unfriendly scenario outlines exactly the situation we want to address with the present work.

Simply put, this contribution can be cast as a mathematical investigation of the physical framework developed in \cite{DCBM17}. We employ different mathematical tools to study the emergence of effective dynamics in scenarios where a system is subject to a unitary evolution and where the effective description of the system is given by a coarse-graining map that takes the original system to an effective Hilbert space of smaller dimension. Our main goal is to develop different tools that give necessary and sufficient conditions for this matching in the general case. To facilitate the reading, we also examine few examples and trace a parallel with the ``classical case''.

Obtaining necessary and sufficient conditions is not a trivial problem, not even in the classical case. Classically, where the systems are described by random variables, as we will see later on, with the help of classical causal inference, it is possible to show that, although emergent, we can infer some sort of correlational influence between the observed macroscopic, coarse-grained variables. Remarkably, this is no longer true in the quantum case, where we have a well-defined effective dynamics only  when  some  sort  of  good matching between the underlying unitary and the coarse-graining holds true. 

We organised the work as follows: Sec.~\ref{Sec.Prel} paves the road for the rest of the paper, as we discuss and define there what we mean by quantum and classical versions of the coarse-graining problem. We envisaged this section to let the reader get traction, and intuition, on the topic. After that all the subsequent sections are dedicated to different approaches to the same problem. We kick it off in Sec.~\ref{Sec.GeoApp} where we show how fiber bundles can help in our problem. Sec.~\ref{Sec.AlgApp}, on the other hand, contains quite the opposite toolbox, as it is dedicated to an algebraic approach. Next, through a semi-definite program we establish a very interesting link between our problem and that of Markovianity in Sec.~\ref{Sec.SDPApp}. Finally, in Sec.~\ref{Sec.UnitEquivApp} inspired on two very standard theorems in quantum information, we show how come there exists a unitary connection between macroscopic and microscopic description in our framework. Wrapping up, we conclude this work discussing what we have done and pointing out to inherent limitations of our formalism and further questions. 

\section{The Coarse-Graining Problem}\label{Sec.Prel}

\subsection{Quantum Case}

Pragmatically, the evolution of closed quantum systems is dictated by a family of unitary channels $\mc{U}_{t}$. If $\rho_0$ is the initial state of this closed system at $t=0$, its evolved state at any time $t>0$ is given by:
\begin{equation}
    \rho_{t}=U_{t}\rho_0U_{t}^{\ast}.
\end{equation}

Nonetheless, as the authors point out in Ref.~\cite{DCBM17}, it might be the case we do not have access to that system unitarily evolving according to $\mc{U}_t$. It might well be the case that either due to imperfections or lack of full knowledge of the underlying system $\rho_t$, we only have access to or can interact with $\sigma_t$, a less informative part of the original system (see fig.~\ref{Fig_Diagram_Quantum}). Even though, we want to ascribe an effective evolution to this system we have in hands. 

Keeping things as simple as possible, but already complicated enough, in ref.~\cite{DCBM17} the authors modelled mathematically this ``lack" or ``loss" of information using a completely positive trace preserving (CPTP) map $\Lambda$ ~\cite{NC00,Watrous18}. Their idea is that whatever these processes are, ultimately they ought to be quantum mechanical. So that, with no way out, they must be represented by a CPTP map acting on a smaller dimensional effective quantum state~\cite{RRCMetal08,PRRFPK18}.
\begin{center}
    \begin{figure}
        \centering
        \includegraphics[scale=0.2]{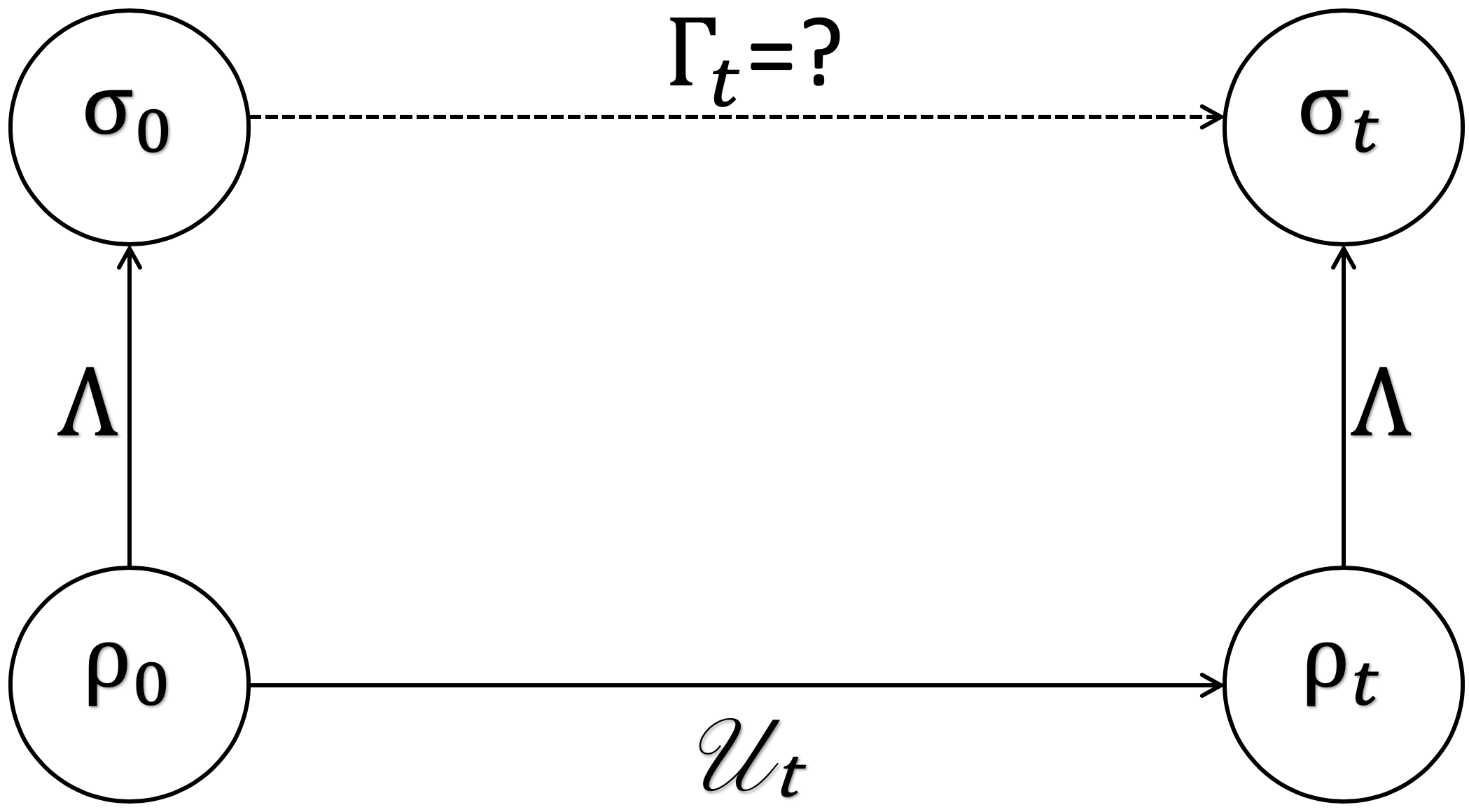}
        \caption{\begin{footnotesize}Coarse-graining diagram. Vertical arrows represent loss, lack or blur of information. Lower horizontal arrow represents the closed, unitary dynamics the system is going through. Uppermost horizontal arrow frames the emergent, perceptible, macroscopic dynamics.\end{footnotesize} }
        \label{Fig_Diagram_Quantum}
    \end{figure}
\end{center}
Although artificial at a first sight, the description adopted in~\cite{DCBM17} turns out to be a mathematically rigorous toy model for coarse-graining. It might be applied to more complex scenarios~\cite{CM19}, and --as we will see later on-- is valid not only in the quantum case, but also in its classical counterpart. 

Simply put, the main question in this framework is:
\begin{center}
\textbf{The Coarse-Graining Problem} \\
    ``\emph{given an underlying unitary dynamics $\mc{U}_{t}$ and a coarse-graining $\Lambda$ acting equally on every instant of time, what are the necessary and sufficient conditions for the existence of an emergent map $\Gamma_{t}$ consistent with such a description, or in other words, that makes the diagram in fig.~\ref{Fig_Diagram_Quantum} commutes.}"
\end{center}

The lowermost level of the diagram in Fig.~\ref{Fig_Diagram_Quantum} represents the microscopic description, whereas the uppermost arrow represents the macroscopic description. We will centre our attention on conditions involving the microscopic underlying dynamics $\mc{U}_{t}$ as well as the coarse-graining $\Lambda$ that give rise to a well-defined, emergent macroscopic dynamics $\Gamma_{t}$. We would like to do so regardless of particular choices of initial states. Our main goal is to find out, then, when it is possible to obtain an emergent dynamics $\Gamma_{t}$ as a map:
\begin{equation}
    \sigma_0 \mapsto \sum_{k=1}^{M}\Gamma_{t,k}\sigma_0\Gamma_{t,k}^{\ast},
\end{equation}
with $\sum_{k=1}^{M}\Gamma_{t,k}^{\ast}\Gamma_{t,k}=\mathds{1}$, for all $t$. 

In this work, we are considering temporal evolution as successive applications of quantum channels. Because of that, we will focus on just one arbitrary snapshot. That is to say that, at least in this paper, we will not be interested in the role played by the time, so that for sake of simplicity we will frequently omit the sub-index $t$ from the maps. To uniform the reasoning, we will use $\rho_0,\sigma_0$ to mean ``initial states" and $\rho_t,\sigma_t$ to mean ``the evolved states at an arbitrary time". 

We conclude this section with an example borrowed from ref.~\cite{SRFCMK15}, as we want to reassure that the abstract setup we are exploring here can in fact be useful for more practical tasks.  In ref.~\cite{SRFCMK15} the authors study a general procedure to classify entangled states via a particular coarse-graining map. Their main idea is to map qu$d$its into qu$b$its, and given a classification of the latter to try and obtain new information for the former. Their proposed dichotomization map is:
\begin{equation}
    \rho_{0} \mapsto \Lambda(\rho_0):=\frac{1}{2}\left(\mathds{1} + \frac{2}{D-1}\sum_{i \in [3]} \langle J_i \rangle_{\rho_0}\hat{\sigma}_{i} \right),
    \label{Eq:CGMapIbrahim}
\end{equation}
where $J_1,J_2,J_3$ are the generators of $SU(2)$ rotations around the $x,y,z-$axes (respec.) in the $D-$dimensional Hilbert space $\mc{H}_D$, and $\hat{\sigma}_{1},\hat{\sigma}_{2},\hat{\sigma}_{3}$ are the Pauli matrices. Notice that their $\Lambda$ is nothing but our coarse-graining map (vertical arrows in Fig.~\ref{Fig_Diagram_Quantum}). Additionally, their underlying unitary quantum dynamics is given by:
\begin{equation}
    \rho_0 \mapsto \rho_{t}:= e^{-i\alpha \langle \vec{J}, \vec{n} \rangle}\rho_0 e^{i\alpha \langle \vec{J}, \vec{n} \rangle},
    \label{Eq:UnitaryMapIbrahim}
\end{equation}
or simply put, a rotation $\alpha$ around a given vector $\vec{n}$. Interestingly enough, for this specific situation there is a compatible --also unitary-- macroscopic emergent quantum dynamics:
\begin{equation}
    \sigma_0 \mapsto \sigma_{t}:= e^{-i\alpha \langle \vec{\hat{\sigma}}, \vec{n} \rangle}\sigma_0 e^{i\alpha \langle \vec{\hat{\sigma}}, \vec{n} \rangle}.
    \label{Eq:EmergentDyamicsIbrahim}
\end{equation}

In the case of ref.~\cite{SRFCMK15}, for commutativity reasons, it is possible to ``interchange" the coarse-graining map with the underlying dynamics. The result of that is (very) a well-defined emergent dynamics. As we will see later on this compatibility is not always true, and we hope our methods can be useful to determine whether or not it is possible to obtain such a behaviour.

\subsection{Classical Case - Classical Inference}

The coarse-graining framework we discuss in this paper can be seen as a mechanism to investigate the emergence of effective quantum dynamics. Simply put, such mechanism is completely inspired by the commutativity expressed in the diagram of fig.~\ref{Fig_Diagram_Quantum}, where the vertical arrows represent loss of information or lack of full access to the underlying quantum system. It is possible, though, to draw a similar diagram for the classical inference problem~\cite{SGS00,Pearl2013Book,Kleinberg2015} and also ask a related question about the emergence of a coarser description. It is exactly this classical counterpart we will examine in this subsection. We hope this gives to the reader something more concrete to deal with before diving deeper into the quantum case and its many variations.

Fig.~\ref{Fig_Diagram_Classical} portrays what we mean by classical (inference) scenario. It leverages Pearl's structural model equation (SME)~\cite{Pearl2013Book} and can be understood as follows: whenever the outcomes of a random variable are determined in terms of the outcomes of another, we draw a directed arrow from the the latter to the former. In the example depicted in fig.~\ref{Fig_Diagram_Classical}, there must be functions $f,g,h$ such that $X(\omega)=f(A(\omega)), B=g(A(\omega))$, and $Y=h(B(\omega))$. Non-deterministic cases are treated likewise, and we refer to the standard reference~\cite{Pearl2013Book} for a mathematically rigorous approach to the subject.

In this mechanistic model, the knowledge of the variables at a deeper scale determines completely the knowledge of the coarser variables. However, as we have been discussing, it is not always true we can have access to the deepest possible level of knowledge. It might be the case that it is only granted access to the coarser information, and even though we must infer what is the influence holding in the set of variables we are dealing with.

At any rate, though, as discussed by the authors in~\cite{Pearl2013Book,Pearl1995,Kleinberg2015,BP16} for classical inference purposes we can get rid of the underlying structural model equation system and look only at the causal diagrams the approach entails. They codify all the necessary information we need to decide about causal statements. These causal diagrams are nothing but directed acyclic graphs (DAGs), and following Pearl's approach, it will be on these diagrams we will centre our attention on.          
\begin{center}
    \begin{figure}
        \centering
        \includegraphics[scale=0.2]{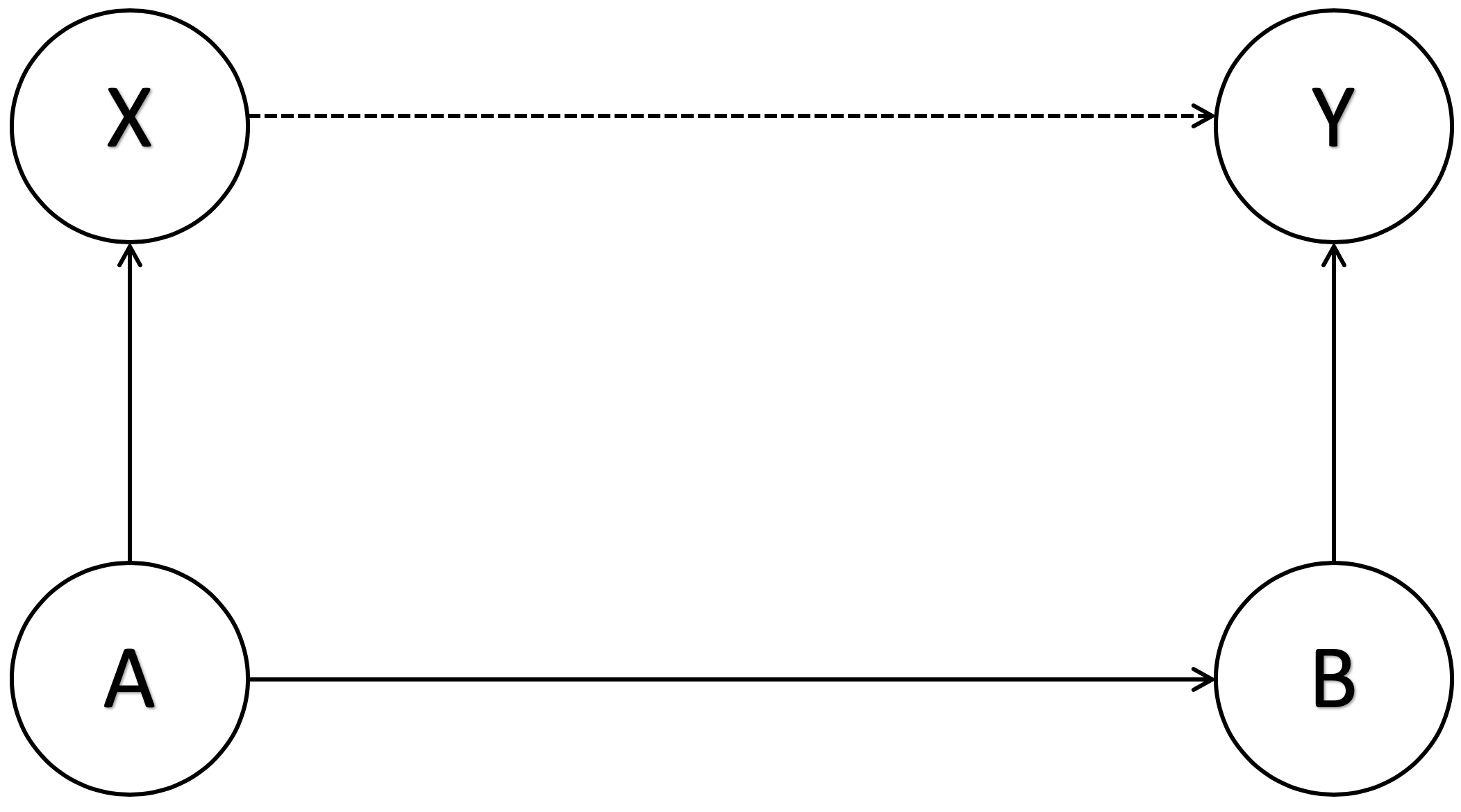}
        \caption{\begin{footnotesize}Classical inference version for the usual quantum coarse-graining scenario. The dotted arrow representing the potential influence $X$ might have on $Y$\end{footnotesize} }
        \label{Fig_Diagram_Classical}
    \end{figure}
\end{center}

As in the quantum case, one can interpret the lower level random variables $A,B$ as some sort of inner property of a system that is not fully available to an external observer, either because they cannot access them completely, or because their measurement apparatus is defective. On the other hand, though, $X$ and $Y$ are the random variables the observer interact with, and gather data from. One can think of them as being proxies for $A$ and $B$. Given that $A$ influences $B$, it is natural to expect that also $X$ influences $Y$ at an emergent level. At a billiard game, for instance, colour is a very rough, coarse-grained description for each ball on the table. Even though, we rely on this very crude description and instead of saying that ``the $10^{23}$ atoms present at a particular region got momentum, evolved, interacted with --transferring momentum to-- other $10^{23}$ atoms on another region making them reach a particular place where they finally disappear", we tend to say ``the white ball knocked off the green ball!". Almost every causal statement, or inference, we make on a daily basis comes from a coarse-grained description of an underlying unattainable system. But how can we transfer this information, making it useful for inferences, from the lower level to the upper most one?

In our naive toy model, the main idea consists of interpreting the information of each random variable in terms of their probability distributions and recognising that influences work like channels, in the very same way we usually do for quantum states and CPTP maps. In this sense, we want to write down $\mathds{P}(Y)$ as function of $\mathds{P}(X)$ in a way that looks like some sort of law of total probability~\cite{GS01} involving also the other probabilities at play. 

As a matter of fact, in clear contrast to the quantum version, we will show that this is always the case: 
\begin{align}
   \mathds{P}(Y)&= \sum_{A,B,X}\mathds{P}(A,B,X,Y)  \nonumber \\
             & = \sum_{A,B,X}\mathds{P}(Y | B)\mathds{P}(B | A) \mathds{P}(X | A)\mathds{P}(A) \nonumber \\
             & = \sum_{A,B,X}\mathds{P}(Y | B)\mathds{P}(B | A) \mathds{P}(A | X)\mathds{P}(X) \nonumber \\
             & = \sum_{X}\tilde{\mathds{P}}(Y | X) \mathds{P}(X),
             \label{Eq:ClassicalApproachInference}
\end{align}
where we have defined
\be
\tilde{\mathds{P}}(Y | X) :=  \sum_{A,B}\mathds{P}(Y|B)\mathds{P}(B | A) \mathds{P}(A | X).
\label{Eq.DefBgivenATilde}
\ee

In summary, eq.~\eqref{Eq:ClassicalApproachInference} says that viewed as an inference problem, the classical version of our quantum scenario always admits an ``arrow" connecting the information contained in $Y$ from $X$ which makes the diagram commutative.

\begin{center}
    \begin{figure}
        \centering
        \includegraphics[scale=0.2]{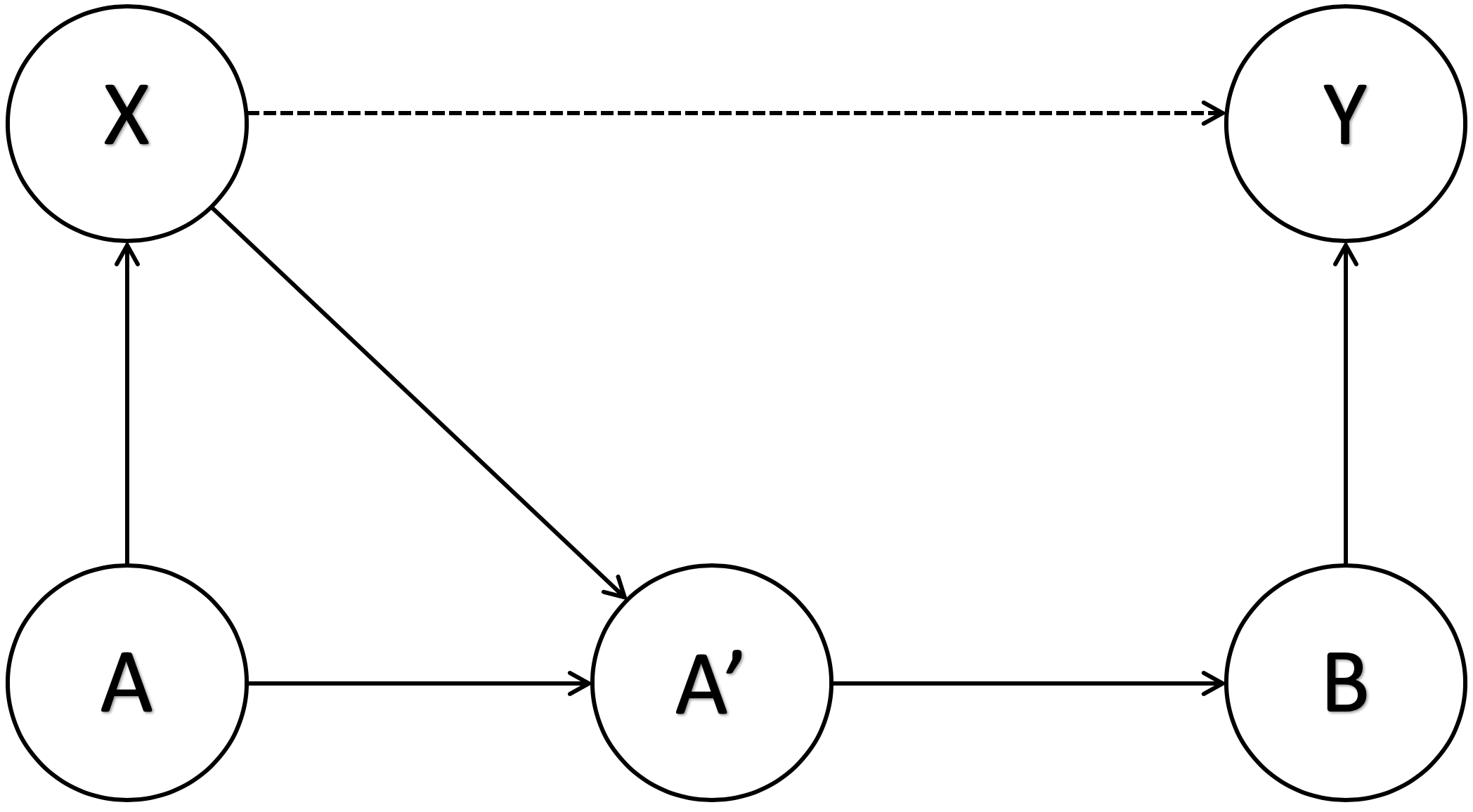}
        \caption{\begin{footnotesize}Slightly modification for the classical coarse-graining diagram. The diagonal arrow means that $X$ also influences $B$. In this picture alterations on $X$ indirectly acts on $B$.\end{footnotesize}}
        \label{Fig_Diagram_Classical_Do}
    \end{figure}
\end{center}

Finally, it is also worth to mention that Pearl's approach is encompassing enough and also provides a slightly different classical analysis of the coarse-graining picture. As a matter of fact, we could have analysed our situation through the lenses of the $do-$calculus~\cite{Pearl2013Book,PM18,Pearl2018,BP16}, and instead of using the language of classical probability channels, we could have studied the influence between $X$ and $Y$ via $\mathds{P}(y | do(X=x))$. We will dedicate this last part of this section for doing so, but we must warn the reader, first, that although inspired on the standard scenario, the $do-$calculus approach we are envisaging here deviates a bit from the standard ones. 

The very first deviation already arises in the diagram we have been work with. In the new diagram (fig.~\ref{Fig_Diagram_Classical_Do}) influences on $B$ are not restricted to the microscopic level alone, as modifications on $X$ are carried out through $A^{\prime}$. This new variable, $A^{\prime}$ is nothing but a proxy for $A$ and has been created only to avoid directed cycles in the causal graph for this scenario. What we are proposing in here intends to naively capture the idea that the macroscopic world also affects the micro-world. 

Other than that, the other point where the approach we are discussing here differs from the previous one is on the usage of the counterfactual $\mathds{P}(Y | do(X=x))$. We refer to ref.~\cite{BP16} for a technical and~\cite{Pearl2018} for a non-technical introduction to the topic. In either case, the main mathematical element we are basing ourselves on is the \emph{$d-$separation} and the \emph{backdoor criterion}~\cite{BP16,Pearl2018}. Putting together, these two results allows to condition on certain variables (variable $B$ in our case) and render the effect of $X$ on $Y$ identifiable: 
\begin{align}
    \mathds{P}&(Y=y | do(X=x)) = \nonumber \\ 
                    &\sum_{b \in \mbox{Out}{B}}\mathds{P}(Y=y | B=b, do(X=x))\mathds{P}(B=b).
                    \label{Eq:ClassicalApproachDoCalculus}
\end{align}

We hope these two classical ways to see the same quantum problem, summarized in Eqs.~\eqref{Eq:ClassicalApproachInference} and~\eqref{Eq:ClassicalApproachDoCalculus}, had given traction to the reader, creating not only an illustrative parallel but also showing fundamental differences between them and the quantum case that we work out in the rest of the paper.


\section{Geometric Approach}\label{Sec.GeoApp}

The geometric viewpoint for the problem starts from the observation that a coarse-graining map is usually a surjective but not injective map $\CG: D_D \rightarrow D_d$.
For each state $\rho \in D_d$ there is a (usually large) set $F\de{\rho} = \DE{\psi\in D_D: \CG \de{\psi} = \rho}$, the fiber over $\rho$~\cite{Hatcher02}.
The existence of effective dynamics on $D_d$ corresponds to the condition that the dynamics on $D_D$ respects such fibering, \ie 
if $\psi$ and $\psi'$ are such that $\CG\de{\psi} = \CG\de{\psi'}$, then \be \CG\de{\Dyn\de{\psi}} = \CG\de{\Dyn\de{\psi'}}.\label{eq:fiber}\ee

\subsection{Example 1}

Consider the coarse-graining $\Lambda$ that takes a three dimensional system to a two dimensional system defined by the  Kraus operators 
\begin{align}
K_0&=\ket{0}\bra{0} + \ket{1}\bra{+}&\\
K_1&=\ket{1}\bra{-}. &
\end{align}
Let us see how the condition \eqref{eq:fiber} constrains the unitaries $U$ acting on the three dimensional system that lead to a well-defined effective dynamics in the two dimensional system.

First we notice that  $U$ must be of the block form
\be U=\begin{bmatrix}
1&0\\
0&U_2\end{bmatrix}\ee
where $U_2$ is a unitary matrix acting on the subspace generated by $\ket{1}$ and $\ket{2}$. Now let us consider the action of $U$ in states of the form
\be \rho=\begin{bmatrix}
1-p&a&b\\
a^*&p\rho_{11}&p\rho_{12}\\
b^*&p\rho_{21}&p\rho_{22}
\end{bmatrix}\ee
writen in the basis $\ket{0}, \ket{+}_{12}, \ket{-}_{12}$, where
\be  \ket{\pm}_{12}=\frac{\ket{1} \pm \ket{2}}{\sqrt{2}}.\ee 
The image of such a state under the action of $\Lambda$ is
\be \begin{bmatrix}
1-p&a\\
a^*&p
\end{bmatrix}.\ee
The image of $\Lambda \circ U$ can not depend on $b$, and this implies that $\ket{\pm}_{12}$ must be eigenvectors of $U_2$,
condition that is also sufficient for $U$ to preserve the equivalence classes defined by $\Lambda$.

\subsection{Example 2}
Consider the coarse-graining $\Lambda$ that takes a $kd$-dimensional system to a $d$-dimensional system such that $k$ levels are mapped to one level of the coarse-grained system, as ilustrated in Fig. \ref{fig:example2}. 

\begin{figure}
    \centering
   \definecolor{yqyqyq}{rgb}{0.5019607843137255,0.5019607843137255,0.5019607843137255}
\definecolor{ududff}{rgb}{0.30196078431372547,0.30196078431372547,1}
\begin{tikzpicture}[scale=3,line cap=round,line join=round,>=triangle 45,x=0.5cm,y=0.5cm]
\draw [line width=1pt] (-8,2.8)-- (-7,2.8);
\draw [line width=1pt] (-8,2.6)-- (-7,2.6);
\draw [line width=1pt] (-8,2.4)-- (-7,2.4);
\draw [line width=1pt] (-8,2)-- (-7,2);
\draw [line width=1pt] (-8,1.8)-- (-7,1.8);
\draw [line width=1pt] (-8,1.6)-- (-7,1.6);
\draw [line width=1pt] (-8,1)-- (-7,1);
\draw [line width=1pt] (-8,0.8)-- (-7,0.8);
\draw [line width=1pt] (-8,0.6)-- (-7,0.6);
\draw [->,line width=1pt,dash pattern=on 1pt off 2pt,color=yqyqyq] (-6.6,2.6) -- (-6,2.6);
\draw [->,line width=1pt,dash pattern=on 1pt off 2pt,color=yqyqyq] (-6.6,1.8) -- (-6,1.8);
\draw [->,line width=1pt,dash pattern=on 1pt off 2pt,color=yqyqyq] (-6.6,0.8) -- (-6,0.8);
\draw [line width=1pt] (-5.6,2.6)-- (-4.6,2.6);
\draw [line width=1pt] (-5.6,1.8)-- (-4.6,1.8);
\draw [line width=1pt] (-5.6,0.8)-- (-4.6,0.8);
\draw [decorate,decoration={brace,amplitude=3pt},xshift=-4pt,yshift=0pt]
(-8,2.3) -- (-8,2.9) node [black,midway,xshift=-0.4cm] 
{\footnotesize $k$};
\begin{scriptsize}
\draw [fill=black] (-7.50363090525873,1.4078825628008227) circle (0.1pt);
\draw [fill=black] (-7.50363090525873,1.3127702929962486) circle (0.1pt);
\draw [fill=black] (-7.50363090525873,1.2118936432035188) circle (0.1pt);
\draw [fill=black] (-5.096793004527191,1.4078825628008227) circle (0.1pt);
\draw [fill=black] (-5.096793004527191,1.3127702929962486) circle (0.1pt);
\draw [fill=black] (-5.096793004527191,1.2118936432035188) circle (0.1pt);
\end{scriptsize}
\end{tikzpicture}
    \caption{Coarse-graining operation where a $kd$-dimensional system is mapped to a $d$-dimensional system such that $k$ levels are mapped to one level of the coarse-grained system.}
    \label{fig:example2}
\end{figure}
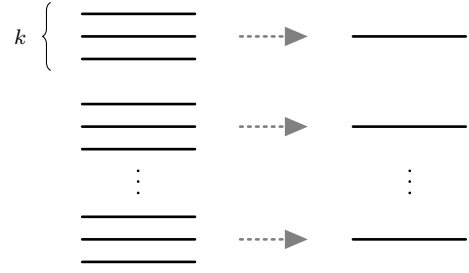

The form of the Kraus maps depend on how coherences are preserved under the action of the coarse-graining map. If coherences are preserved, we have $k$ Kraus maps of the form
    \be K_i=\sum_{j=0}^{d-1}\ket{j}\bra{u_{ij}}.\ee The vectors $\ket{u_{ij}}$ depend on how coherences transform.

If coherences are not preserved, then we have $kd$ Kraus operators of the form 
  \be K_{ij}=\ket{i}\bra{u_{ij}}.\ee
  
  If some coherences are preserved but not others, we will have something intermediate between one and $d$ Kraus operators for each $k$ level subsystem of dimension $k$ on the original system.
  
  In order to preserve the equivalence classes defined by $\Lambda$, a unitary $U$ must be of block form
  \be U=\begin{bmatrix}
  U_0&0&\ldots &0\\ 
  0&U_1&\ldots &0\\
  \vdots &\vdots &\vdots &\vdots \\
  0&0&\ldots &U_{d-1} 
  \end{bmatrix}.\label{eq:blockunitary}\ee
  This condition is necessary but not sufficient. Extra constrains on the blocks $U_i$ depend on the number of Kraus operators and the vectors $\ket{u_{ij}}$, as we show in the particular case below.

Consider the coarse-graining $\Lambda$ of a $4$-dimensional Hilbert space with $k=2$ given by the Kraus operators
\be K_0 = \ket{0}\bra{+}_{01} +\ket{1}\bra{+}_{23}\ee 
\be K_1 = \ket{0}\bra{-}_{01} +\ket{1}\bra{-}_{23}\ee
with $\ket{\pm}_{01}=\frac{\ket{0}\pm\ket{1}}{\sqrt{2}}$ and $\ket{\pm}_{23}=\frac{\ket{2}\pm\ket{3}}{\sqrt{2}}$.

Writing $\rho$ using the basis $\left\{\ket{\pm}_{01},\ket{\pm}_{23}\right\}$  we have that 
\be \Lambda\left(\rho\right)=\begin{bmatrix}\rho_{00}+\rho_{11}&\rho_{02}+\rho_{13}\\ \rho_{20}+\rho_{31}&\rho_{11}+\rho_{33}\end{bmatrix}\ee
Let $\rho$ and $\sigma$ be matrices such that $\Lambda\left(\rho\right)=\Lambda\left(\sigma\right)$. Taking $U$ of the form \eqref{eq:blockunitary}, we have that $\Lambda\left(U\rho U^\ast\right)=\Lambda\left(U\sigma U^\ast\right)$ if and only if 
\be \tr \left(U_1C_{\rho}U_2^\ast\right)=\tr\left(U_1C_{\sigma}U_2^\ast\right).\ee
 where \be C=\begin{bmatrix}\rho_{02}&\rho_{03}\\
\rho_{12}&\rho_{13}\end{bmatrix}\ee is the coherence matrix between states $\ket{\pm}_{01}$ and $\ket{\pm}_{23}$ for matrix $\rho$ and similar for $C_{\sigma}$. This in turn implies that the matrix form of $U_1$ with respect to the basis $\ket{\pm}_{01}$ and the matrix form of $U_2$ with respect to the basis $\ket{\pm}_{23}$ must be equal.


\section{Algebraic Approach}\label{Sec.AlgApp}

We dedicate this section for an algebraic-like approach to our main problem. Our principal question remains the same: we would like to find out conditions for the existence of the effective map $\Gamma$ connecting the upper level of the diagram in fig.\ref{Fig_Diagram_Quantum}. Although connected with 
the previous section, the results we show here may be regarded as being a bit more explicit than those we have provided through geometrical arguments.

To begin with, we will first obtain a condition guaranteeing the existence of such map as a well-defined function~\cite{DSW18,DCBM17}.  

\begin{thm}
$\Gamma_{t}$ is a well-defined function if, and only if, 
\begin{equation}
\Lambda(\mathcal{U}_t\rho_0)= \Lambda(\mathcal{U}_t\tilde{\rho}_0),
\label{Eq.ConditionOnU}
\end{equation}
whenever $\Lambda(\rho_0)=\Lambda(\tilde{\rho}_0)$.
\label{Thm.PreserveEquivalenceClass}
\end{thm}
\begin{proof}
$(\Rightarrow)$ 
Suppose that there exists a well-defined function $\Gamma_{t}$ making the diagram commutative. In this case, whenever we take $\rho_0,\tilde{\rho}_0 \in \mathcal{H}_{D}$ satisfying $\Lambda(\rho_0)=\sigma_0=\Lambda(\tilde{\rho}_0)$ we know that
\begin{equation}
    \Gamma_t(\sigma_0)=\Gamma_t(\Lambda(\rho_0))=\Lambda(\mathcal{U}_t(\rho_0))
    \label{Eq.EqThmPreserveCond1}
\end{equation}
and also that
\begin{equation}
    \Gamma_t(\sigma_0)=\Gamma_t(\Lambda(\tilde{\rho}_0))=\Lambda(\mathcal{U}_t(\tilde{\rho}_0)).
    \label{Eq.EqThmPreserveCond2}
\end{equation}
Therefore,  $\Lambda(\mathcal{U}_t\rho_0)= \Lambda(\mathcal{U}_t\tilde{\rho}_0)$.

$(\Leftarrow)$
Now, let us assume that $\Lambda(\mathcal{U}_t\rho_0)= \Lambda(\mathcal{U}_t\tilde{\rho}_0),$ when $\Lambda(\rho_0)=\Lambda(\tilde{\rho}_0)$. We must show that for each vector in $\Lambda(\mathcal{H}_{D}) \subset \mathcal{H}_{d}$ there is only one vector lying also in $\mathcal{H}_d$, assigned by $\Gamma_t$, that makes the diagram consistent. For each $\sigma_0$ define the following assignment:
\begin{equation}
\sigma_0 \mapsto (\Lambda \circ \mathcal{U}_t) \circ [\Lambda^{-1}(\sigma_0)].
\label{Eq.ThmAssociation}
\end{equation}
We affirm that the above mapping does define a function making the diagram commutative. Otherwise, it would exist one $\sigma_0$ associated with two different vectors through Eq.~\eqref{Eq.ThmAssociation}, and in addition it would also exist at least $\rho_0 \neq \tilde{\rho}_0$ belonging to $\Lambda^{-1}(\sigma_0)$ such that $\Lambda(\mathcal{U}_t \rho_0)=\Lambda(\mathcal{U}_t\tilde{\rho}_0)$, as both $\rho_0$ and $\tilde{\rho}_0$ belong to the inverse image of $\Lambda$ and therefore preserved by $\mathcal{U}_t$, an absurd.
\end{proof}

Remarkably, whenever we have $\Gamma_t$ as a well defined function, consistent with the diagram, it will also be linear and positive. It suffices to follow the bottom part of the diagram to see that it holds true:
\begin{align}
\Gamma_t(a\sigma_0 &+ b\tilde{\sigma}_0)= \nonumber \\ &= \Gamma_t(a\Lambda(\rho_0)+b\Lambda(\tilde{\rho}_0)) \nonumber \\
& =\Gamma_t(\Lambda(a\rho_0+b\tilde{\rho}_0)) \nonumber \\
& = \Lambda(\mathcal{U}_t(a\rho_0+b\tilde{\rho}_0)) \nonumber \\
&=a\Lambda \circ \mathcal{U}_t(\rho_0) + b\Lambda \circ \mathcal{U}_t(\tilde{\rho}_0) \nonumber \\
&= a\Gamma_{t}(\sigma_0) + b\Gamma_{t}(\tilde{\sigma}_0), \,\, a,b \in \mathds{C}
\end{align}
and

\begin{align}
    \Gamma_t(\sigma_0) &=\Gamma_{t}(\Lambda(\rho_0)) \nonumber \\
    &=(\Lambda \circ \mathcal{U}_t)(\rho_0) \geq 0, \,\, \forall \sigma \in \Lambda(\mathcal{D}(\mathcal{H}_D)). 
\end{align}
In plain English, thm.~\ref{Thm.PreserveEquivalenceClass} says that the existence of an well-defined effective dynamics is guaranteed when there exists some sort of good matching between the underlying unitary $\mathcal{U}_t$ and the coarse-graining $\Lambda$. More precisely, whenever the unitary preserves --in the sense of Eq.~\eqref{Eq.ConditionOnU}-- inverse images, a linear and positive effective dynamics might be defined. 

Exploring a bit further this idea of preservation, or in other words this well-formed-matching between $\Lambda$ and $\mc{U}_{t}$, we could ask whether there is some more meaningful condition also implying the existence of an effective dynamics. The answer of this question is at the heart of the following theorem:

\begin{thm}
Let $\{M_k\}_{k=1}^{N}$ be the Kraus decomposition for the coarse-graining $\Lambda$, \emph{i.e.}
\begin{equation}
    \Lambda(\rho)=\sum_{k \in [N]}M_k \rho M_{k}^{\ast}, \,\, \forall \rho \in \mathcal{D}(\mathcal{H}_D)
\end{equation}
If there exists $V \in \mathds{M}_{d \times d}$ such that 
\begin{equation}
    M_kU_t=VM_k, \,\, \mbox{for all} \,\, k \in [N]
    \label{Eq.ThmCommutativity}
\end{equation}
then it is possible to define an effective map $\Gamma_t$ making the diagram commutative.
\label{Thm.CommutativityUandVandM}
\end{thm}
\begin{proof}
Using thm.~\ref{Thm.PreserveEquivalenceClass}, we know it suffices to prove that $\mc{U}_t$ preserves the inverse image of $\Lambda$. Let, then, $\rho_0,\tilde{\rho}_0 \in \Lambda^{-1}(\sigma_0)$:
\begin{align}
    \Lambda(\mc{U}_t(\rho_0))&=\sum_{k \in [N]}M_kU_t\rho_0U_{t}^{\ast}M_{k}^{\ast} \nonumber \\
    &=\sum_{k \in [N]}VM_k\rho_0M_{k}^{\ast}V \nonumber \\
    &=\sum_{k \in [N]}VM_k\tilde{\rho}_0M_{k}^{\ast}V \nonumber \\
    &=\sum_{k \in [N]}M_kU_t\rho_0U_{t}^{\ast}M_{k}^{\ast} \nonumber \\
    &= \Lambda(\mc{U}_t(\tilde{\rho}_0))
    \end{align}
\end{proof}
We should point out, though, that whenever the commutativity expressed in eq.~\eqref{Eq.ThmCommutativity} holds true, we may go even further and by recalling that the Kraus decomposition of a CPTP map satisfies $\sum_{k=1}^{N}M_{k}^{\ast}M_{k}=\mathds{1}$ we end up getting:
\begin{align}
    M_kU_t =& VM_k, \,\, \forall k \nonumber \\ 
    & \Longrightarrow M_{k}^{\ast}M_kU_t=M_{k}^{\ast}VM_k, \,\, \forall k \nonumber \\
    &\Longrightarrow \sum_{k \in [N]}M_{k}^{\ast}M_kU_t=\sum_{k \in [N]} M_{k}^{\ast}VM_k \nonumber \\
    &\Longrightarrow U_t = \sum_{k \in [N]} M_{k}^{\ast}VM_k.
    \label{Eq.ReinterpretingCommCondition}
\end{align}

In conclusion, whenever eq.~\eqref{Eq.ThmCommutativity} is valid, it is possible to interpret the content of its relation as saying that we can see the underlying unitary as arising from the application of the dual $\Lambda^{\ast}$ on an fictitious matrix $V$ in $\mathds{V}_{d \times d}$ acting on the upper-most level. Remarkably, this result shows that the action of the dual also holds relevant physical information about the emergent dynamics scenario.


\section{SDP Approach}\label{Sec.SDPApp}

In this section we explore a third different point of view that is useful to come up with necessary and sufficient conditions for the existence --in general-- of an effective map $\Gamma_{t}$ consistent with the diagram contained in fig.~\ref{Fig_Diagram_Quantum}. 

The strategy we implement here employs the connection~\cite{BD16,BMA17,LHW18,RBY18,CRE18,GBV18} between divisibility of quantum dynamical maps and non-increasing of information. More precisely, we show how our scenario can be understood within the usual framework of deciding whether or not a given dynamics is Markovian~\cite{RHP14}, and then by using the idea of a \emph{completely information decreasing map}~\cite{BD16} we will also show how the non-emergence of an effective map can be witnessed through a violation of an information-theoretical inequality.

We start with the following result from ref.~\cite{BD16}:

\begin{thm}
A given discrete-time quantum dynamical map, represented by a family of CPTP maps
\be
\mc{N}_{i}: \mc{D}(\mc{H}_{D}) \longrightarrow \mc{D}(\mc{H}_{d_i}), \,\, i \in \,\, [T],
\ee
is divisible if, and only if, for any auxiliary Hilbert space $\mc{H}_{N}$ and for any finite ensemble $\mc{E}=\{p_x,\rho_{DN}^{x}\}_{x}$ in $\mc{H}_{D} \otimes \mc{H}_{N}$ the following chain of inequalities hold true:
\be 
P_{\mbox{guess}}(\mc{E}_1) \geq P_{\mbox{guess}} (\mc{E}_{2}) \geq ...  \geq P_{\mbox{guess}} (\mc{E}_{T}),
\label{Eq.OriginalIneqsBuscemiDatta}
\ee
\label{Thm.BuscemiDatta}
where 
\be
P_{\mbox{guess}}(\mc{E}_{i}):= \max_{\{\Pi_{x}\}}\sum_{x}p_{x}Tr[\Pi_{x}(\mc{N}_{i} \otimes \mbox{id})(\rho_{DN}^{x})],
\ee
with maximum being taken over all POVM's acting on the product space $\mc{H}_{D} \otimes \mc{H}_{N}$.
\end{thm}

\begin{cor}
Let $\mc{U}_{t}$ be a unitary map acting on $\mc{D}(\mc{H}_{D})$ and $\Lambda: \mc{D}(\mc{H}_{D}) \longrightarrow \mc{D}(\mc{H}_{d})$ be a CPTP map representing the effect of a coarse-graining. There exists an emergent dynamics described by a CPTP map 
\be
\Gamma_{t}: \mc{D}(\mc{H}_{d}) \longrightarrow \mc{D}(\mc{H}_{d})
\ee
if, and only if, for all finite-dimensional Hilbert spaces $\mc{H}_{N}$ and for all ensembles $\{p_{x},\rho_{DN}^{x}\}_{x}$ in $\mc{H}_{D} \otimes \mc{H}_{N}$ the following inequality holds true
\begin{align}
\max_{\{\Pi_{x}\}}&\sum_{x}p_x \mbox{Tr}[\Pi_{x}(\Lambda \otimes \mbox{id})(\rho_{DN}^{x})] \geq \nonumber \\
&\geq \max_{\{\Pi_{x}\}}\sum_{x}p_x \mbox{Tr}[\Pi_{x}(\Lambda \otimes \mbox{id})(\mc{U}_{t} \otimes \mbox{id})(\rho_{DN}^{x})],
\label{Eq.CoroIneqBuscemiOurCase}
\end{align}
where the maximum is taken over the set of all POVM's acting on $\mc{H}_{d} \otimes \mc{H}_{N}$.
\label{Cor.BuscemiAndDattaInOurCase}
\end{cor}

Considering the action of the dual maps $\Lambda^{\ast}$ and $\mc{U}_{t}^{\ast}$, and also defining 
\be
\tilde{\Pi}_{x}:= (\Lambda^{\ast} \otimes \mbox{id})(\Pi_{x}), \forall \,\, x
\label{Eq.LambdaStarActingOnPOVMs}
\ee
we could rewrite the content of ineq.~\eqref{Eq.CoroIneqBuscemiOurCase} in a more meaningful way: 
\begin{align}
&\max_{\{\tilde{\Pi}_{x}\}}\sum_{x}p_x \mbox{Tr}[\tilde{\Pi}_{x}\rho_{DN}^{x}] \geq \nonumber \\
&\geq \max_{\{\tilde{\Pi}_{x}\}}\sum_{x}p_x \mbox{Tr}[(\mc{U}_{t}^{\ast} \otimes \mbox{id})(\tilde{\Pi}_{x})\rho_{DN}^{x}] \nonumber \\
&= \max_{\{\tilde{\Pi}_{x}\}}\sum_{x}p_x \mbox{Tr}[({U}_{t}^{\ast} \otimes \mbox{id})(\tilde{\Pi}_{x})({U}_{t} \otimes \mbox{id})\rho_{DN}^{x}] \nonumber \\
&= \max_{\{\tilde{\Pi}_{x}\}}\sum_{x}p_x \mbox{Tr}[(\tilde{\Pi}_{x}^{t})\rho_{DN}^{x}],
\label{Eq.CoroIneqBuscemiOurCase2ndVersion}
\end{align}
with $\tilde{\Pi}_{x}^{t}:=(U_{t}^{\ast} \otimes \ \mbox{id})\tilde{\Pi}_{x} (U_{t} \otimes \ \mbox{id})$ for all $x$. This second, more visual version of corollary~\ref{Cor.BuscemiAndDattaInOurCase} allows us to formulate the following result:
\begin{cor}
Let $\Lambda:\mc{H}_{D} \longrightarrow \mc{H}_{d}$ be a CPTP map, and let 
\begin{align}
\Pi&:=\left\lbrace \,\, \{(\Lambda^{\ast} \otimes \mbox{id})(\Pi_{x})\}_{x};  \,\, \{\Pi_{x}\} \,\, \mbox{is}\right. \nonumber \\
&\left. \mbox{a POVM acting on} \,\, \mc{H}_{D} \otimes \mc{H}_{d} \right\rbrace,
\label{Eq.DefiningPOVMCoro}
\end{align} 
so that whenever 
\be
(\mc{U}^{\ast}_{t}\otimes id)(\Pi) \subset \Pi
\label{Eq.CorPreservingPOVMs}
\ee
there is a well-defined CPTP map $\Gamma_{t}$ making the diagram commutative.
\label{Cor.PreservationOfPOVM}
\end{cor}
\textit{Remark:} Note that the set $\Pi$ appearing in eq.~\eqref{Eq.DefiningPOVMCoro} is a subset of the set of all POVMs acting on $\mc{D}(\mc{H}_{D}) \otimes\mc{D}(\mc{H}_{N})$ generated by the application of $\Lambda^{\ast} \otimes \mbox{id}$. As a matter of fact, 
as $\Lambda$ is a CPTP map, its dual $\Lambda^{\ast}$ is unital, and as such it maps $\mathbb{1}_{D}$ the identity in $\mc{L}(\mc{H}_{d})$ onto $\mathbb{1}_{D}$ the identity in $\mathcal{L}(\mathcal{H}_{D})$. This guarantees that whenever $\{\Pi_{x}\}_{x}$ is a POVM, $\{(\Lambda^{\ast} \otimes \mbox{id})(\Pi_{x})\}_{x}$ also is a valid POVM, acting on a different space though.

In a nutshell, we can interpret corollary~\ref{Cor.PreservationOfPOVM} above as saying that whenever the underlying unitary $\mc{U}_{t}$ preserves the structure of the measurements arising from the application of the coarse-graining on the set of all POVMs acting on the product $\mc{H}_{D} \otimes \mc{H}_{N}$, it is possible to define an emergent effective dynamics consistent with the diagram. As expected, this result shows that to obtain $\Gamma_{t}$ as a proper quantum channel, we should enforce some good matching between $\mc{U}_{t}$ and $\Lambda$.


\section{Unitary Equivalence}\label{Sec.UnitEquivApp}

We conclude our list of different approaches to the coarse-graining scenario with what we call the \emph{unitary equivalence} approach. Despite being rather abstract, and difficult to be checked in practice even in the simpler cases~\cite{DYW19}, the necessary and sufficient condition we come up with and discuss here shall be seen as a way to study properties of the possible emergent macroscopic dynamics as function of the underlying microscopic evolution.

To a certain extent, the method we work here here is, to a certain extent,  a consequence of the well-known equivalence of ensembles~\cite{Watrous18}:
\begin{thm}
Let $\{\psi_{i}\}_{i=1}^{N}$ and $\{\phi_{i}\}_{i=1}^{N}$ be two set of non-necessarily normalized vectors in $\mathcal{H}$. The equality
\begin{equation}
    \sum_{i \in [N]}\ketbra{\psi_i}{\psi_i} = \sum_{i \in [N]}\ketbra{\phi_i}{\phi_i}
    \label{Eq:ThmEqMixture}
\end{equation}
holds true if, and and only if, there is an $N \times N$ unitary matrix $U$ satisfying
\begin{equation}
    \ket{\psi_i}=\sum_{j \in [N]}U_{ij}\ket{\phi_{j}}, \,\, \forall i \in [N].
\end{equation}
\label{Thm:EnsambleEquivalence}
\end{thm}

Remarkably, thm.~\ref{Thm:EnsambleEquivalence} above says that there is no way out, whenever two ensembles are the same there must be a tight connection linking them together in a manner that it is possible to write down the constituents of one in terms of the other. Being possible, therefore, to analyse properties of one looking at known properties of the other. It is this very same idea we will transport to our problem, and explore from now on in this section. 

We start with a rather important corollary of thm.~\ref{Thm:EnsambleEquivalence}:
\begin{cor}
Two set of Kraus operators $\{K_i\}_{i=1}^{N}$ and $\{\tilde{K}_{j}\}_{j=1}^{N}$ represent the same quantum channel if, and only if, there is an $N \times N$ unitary $U$ satisfying 
\begin{equation}
K_i=\sum_{j \in [N]}U_{ij}\tilde{K}_{j}
\label{Eq:EqCoroEquivalentChannel}
\end{equation}
\label{Cor:EquivalenceOfQuantumChannels}
\end{cor}

We refer to ref.~\cite{Watrous18} for a proof of this result. In a nutshell, though, the main aspects of the proof goes as follows: first, a translation of the problem from equivalent quantum channels to equivalent states, via Choi-Jamiolkoski isomorphism; and second, the use of thm.~\ref{Thm:EnsambleEquivalence} for necessary and sufficient conditions for two ensembles. Much in the spirit of thm.~\ref{Thm:EnsambleEquivalence}, the corollary~\ref{Cor:EquivalenceOfQuantumChannels} says that whenever two quantum channels are equal, there must be a tight connection -- represented by the unitary $U$ in Eq.~\eqref{Eq:EqCoroEquivalentChannel} -- between any ``different" Kraus representation of such a channel.  

Now we have got in hands all the necessary ingredients to establish our desired connection between the microscopic and the macroscopic level of our coarse-graining diagram (fig.~\ref{Fig_Diagram_Quantum}). To begin with, given $\Lambda$ a coarse-graining map with set Kraus operators being $\{M_k\}_{k=1}^{N}$, suppose that there exists an emergent $CPTP$ dynamics 
\begin{align}
    \Gamma_{t}(\sigma_0)=\sum_{j \in [M]}K_j \sigma_0 K_{j}^{\ast}
    \label{Eq.KrausMapsForGamma}
\end{align}
making the diagram commutes. In this situation, for all $\rho_0 \in \mc{D}(\mc{H}_D)$ it must be true that:
\begin{align}
    (\Gamma_{t} \circ \Lambda)(\rho_0) & = (\Lambda \circ \mc{U}_{t})(\rho_0) \nonumber \\
    \sum_{i,j} K_i M_{j}\rho_0 M_{j}^{\ast}K_{i}^{\ast} & = \sum_{i}M_{i}U \rho_{0} U^{\ast}  M_{i}^{\ast},
    \label{Eq:ChainEqEquivalent}
\end{align}
and therefore the two set of Kraus maps for the composition $\{K_iM_j\}_{(i,j) \in [M]\times[N]}$ and $\{M_i U\}_{i \in [N]}$ must be unitary equivalent. Which means that there must exist a unitary matrix $V$ satisfying:
\begin{equation}
    K_i M_j = \sum_{k \in [N]}V_{(i,j),k} M_{k}U
\end{equation}
for all $i,j$ in $[N] \times [M]$. 

Moving on, on the other hand now, suppose we are given not only a coarse-graining map $\Lambda$ with Kraus operators $\{M_k\}_{k=1}^{N}$ and the underlying dynamics $\mc{U}$, but also a fixed unitary~\footnote{As we can add as many zeroes we want, we are deliberately omitting the dimension of the unitary matrices in our argument.} matrix $W$. In these circumstances it is possible to define a set of Kraus operators $\{K_{i,j}\}_{i,j}$ given by
\begin{equation}
    K_{i,j}:=\sum_{k \in [N]}W_{(i,j),k}M_kU, \,\, \mbox{for all} \,\, i,j,
    \label{Eq:DefiningGammaViaEquivalenceKraus}
\end{equation}
as well as a map from $\Lambda(\mc{D}(\mc{H}_D))$ onto itself given by
\begin{equation}
    \Gamma(\Lambda(\rho_0)):=\sum_{i,j,k}K_{(i,j),k}\rho_0K_{(i,j),k}^{\ast}.
    \label{Eq:DefGammaViaEquivalenceMap}
\end{equation}
Because eq.~\eqref{Eq:DefiningGammaViaEquivalenceKraus} unitarily connects $K_{i,j}$'s with $M_k$'s, corollary~\ref{Cor:EquivalenceOfQuantumChannels} says that the uppermost and the lowermost paths in the coarse-graining diagram must coincide. In other words, the emergent map expressed by the composition in Eq.~\eqref{Eq:DefGammaViaEquivalenceMap} makes the diagram commutes. In other words, we have just proven the following result:
\begin{prop}
There exists $\Gamma_{t}$, with Kraus operators $\{K_i\}_{i=1}^{M}$ making the diagram commutes if, and only if, there is a unitary matrix $V$ satisfying $K_iM_j=\sum_{k \in [M]}V_{(i,j),k}M_kU$, where $U$ represents the underlying unitary dynamics and $\{M_j\}_{j=1}^{N}$ is the set of Kraus operators for the coarse-graining map $\Lambda$
\end{prop}
We summarise the last paragraphs above saying that to have a well-defined emergent macroscopic dynamics, the Kraus decomposition of it must be unitarily connected with both the underlying unitary closed dynamics and the coarse-graining map as expressed in eq.~\eqref{Eq:EqCoroEquivalentChannel}, regardless of the initial states in $\mc{D}(\mc{H}_D)$ or subsets of $\Lambda(\mc{D}(\mc{H}_D))$ -- as noticed by the authors in ref.~\cite{DCBM17}. That is to say that, although very abstract, our result shows that it is possible to study properties of the emergent map examining, instead, the coarse-graining and the unitary underlying dynamics. On the other way round, let us mention that in the situation in which the emergent $\Gamma_t$ and the coarse-graining maps are known, it is also possible to infer features of the underlying dynamics the microscopic system is going through.


\section{Discussion}\label{Sec.Disc}

This contribution compiles four different approaches to the study of coarse-grained quantum dynamics emerging from loss or lack of information. In our coarse-graining scenario, the inherent imperfections manifested in the macroscopic level are modelled by the action of a CPTP map. Microscopically, where all the information is suppose to be encoded, the system we are interested in evolves unitarily. The existence of an emergent dynamics equals the commutativity of the diagram in fig.~\ref{Fig_Diagram_Quantum}. To tackle this problem, we have shown how it is possible to use geometrical, set theoretical, optimisation and linear-algebraic arguments. Each of which having their benefits, practical limitations, particularities and similarities. All of them showing that the commutativity asked for the well-definition of the emergent dynamics can be expressed in terms of a good matching between the coarse-graining map and the unitary evolution. In particular, our geometrical construction exemplifies this characteristic quite clearly. The two examples we construct there aim to show how much the coarse-graining maps constrict the possible family of unitary dynamics that would give rise to a well-defined emergent map.

We have also made a classical digression, comparing the usual quantum scenario with two possible classical ones. Differently from the quantum case, both classical approaches say that it is always true that, despite emergent, we can infer some sort of influence between the observed coarse-grained variables. Additionally, we also find this case useful due to its proximity with the conditional quantum state approach~\cite{LS13,LS14}. In the latter the authors show that there is always a translation, at least on the structural level, from equations involving classical probability to equations involving quantum conditional states. This opens up another venue of investigation that, we hope, will be further explored in more details in subsequent works. 

Still focusing on the classical case, we must mention that even though our claim is correct, there is something not completely satisfactory about it. In the definition of $\tilde{\mathds{P}}(Y | X )$ we made use of $\mathds{P}(A | X)$. In a certain sense that probability goes towards the wrong direction, as the direction of the arrow connecting $A$ and $X$ starts at $A$ and ends at $X$. Whereas it does not represent any obstacle from the mathematical point of view, as it always possible to do so, from a physical perspective it blocks off better interpretations of our findings, as it is rather difficult to get a hold of it if we only can access $X$ from $A$ and not the other way round. We hope the analysis through quantum conditional states, as we mentioned, together with a novel notion of arrow reversion can shed new light on this problem.

We would like to conclude this paper discussing prospective and future works. We have already pointed out to the possible fruitful connection between conditional quantum states and the coarse-graining scenario: be it to strengthen the parallel between the classical and quantum cases or be it to investigate different and foundational definitions of good-matching for the diagram of fig.~\ref{Fig_Diagram_Quantum}, as it has been done in ref.~\cite{Duarte20}. Another dimension that can also be explored ties back to the SDP characterisation of the problem, as we think it is also possible to write down the emergent map (or the Choi-Jamio\l kowski image of it) whenever such map happens to exist. Writing down the emergent map is important not only within the scope of our problem but also for studying memory effects in open quantum systems. Finally, we would like to mention that it might also be possible that our diagrams express matching relations much in the same spirit as the compatibility relations present in refs.~\cite{HM17, CHT19, Mori20, HMZ16}. That would bring in not only another different -- and well-developed -- mathematical toolbox but also a completely different physical perspective to both problems. We hope to bridge this gap in the near future.


\begin{acknowledgments}
CD wishes to thank F. de Melo for introducing him to the subject. The authors thank  J. Piennar and B. Rizzuti for 
valuable discussions.\\

BA acknowledges financial support from the Brazilian ministries and agencies MEC, 
MCTIC, and FAPEMIG. MTC and BA acknowledges financial support from CNPq. CD was supported by a fellowship from the Grand Challenges Initiative at Chapman University.\\

This paper is part of the Brazilian National Institute on Science and Technology on Quantum Information - INCT-IQ This project/research was supported by grant number FQXi-RFP-IPW-1905 from the Foundational Questions Institute and Fetzer Franklin Fund, a donor advised fund of Silicon Valley Community Foundation.
\end{acknowledgments}

\bibliography{biblio3}
\end{document}